\documentclass[a4paper,11pt,reqno]{article}
\usepackage[a4paper, total={6in, 8in}]{geometry}
\usepackage{enumerate}
\usepackage{amsmath}
\usepackage{amssymb,latexsym}
\usepackage{amsthm}
\usepackage{color}
\usepackage{cancel}
\usepackage{graphicx}
\usepackage{cite}
\usepackage{changes}
\usepackage{lineno}
\usepackage{hyperref}
\usepackage{comment}
\usepackage{amscd}

\newtheorem{theorem}{Theorem}[section]
\newtheorem{problem}[theorem]{Problem}

\newtheorem{definition}[theorem]{Definition}
\newtheorem{proposition}[theorem]{Proposition}

\newtheorem{remark}[theorem]{Remark}

\newcommand{\F}{{\mathbb F}}

\newcommand{\fq}{{\mathbb F}_{q}}

\title{Left ideal LRPC codes and a ROLLO-type cryptosystem based on group algebras}
\author{Martino Borello, Paolo Santonastaso and Ferdinando Zullo}
\date{ }

\begin{document}
\maketitle

\begin{abstract}
In this paper we introduce left ideal low-rank parity-check codes by using group algebras and we finally use them to extend ROLLO-I KEM.
\end{abstract}


\section{Introduction}

Low-rank parity-check (LRPC) codes are rank-metric codes introduced by Gaborit et al. in \cite{aragon2019low,gaborit2013low} as the rank-metric analogs of low-density parity-check codes in the Hamming metric. These codes turned out to be very interesting for their efficient probabilistic decoding algorithm and their applications in several contexts, such as powerline communications \cite{yazbek2017low}, network coding \cite{el2018efficient} and cryptography \cite{melchor2020rollo}.
Regarding the last one, in ROLLO cryptosystem special types of LRPC codes are used, known as ideal LRPC codes, which can be efficiently stored: they can be simply represented by a vector instead of the entire generator matrix (as in the case with circulant matrices).

Motivated by the application in cryptography, in this paper we introduce left ideal LRPC codes arising from group algebras, deriving some interesting properties, such as a systematic parity-check matrix. Thank to this latter property, we are able to describe a new variant of the Key-Encapsulation scheme ROLLO-I with the use of group algebras. A similar approach has been recently used in \cite{chimal2022efficient} in the Hamming-metric context (MDPC codes), to generalize the BIKE cryptosystems.

\medskip

Throughout the paper, $q$ is a prime power, $m$  is a positive integer, $\F_{q^m}$ is the field with $q^m$ elements and $G$ is a finite group of size $n$.

\section{Group algebras and left ideal matrices}

The group algebra $\F_{q^m} G$ is the set 
$\left\{\sum_{g\in G} a_g g \colon a_g \in \F_{q^m} \right\}$,
which is an $\F_{q^m}$-vector space of dimension $n$ in a natural way and it is also an $\F_{q^m}$-algebra via the multiplication 
\[
ab:=\sum_{g \in G}\left(\sum_{h \in G} a_h b_{h^{-1}g} \right)g,
\]
for $a=\sum_{g \in G}a_g g$ and $b=\sum_{g \in G}b_g g$. From now, \emph{we fix an ordering} $1_G=g_1,\ldots,g_n$ of the elements of $G$. The $\F_{q^m}$-vector spaces  $\F_{q^m}^n$ and $\F_{q^m}G$ may be identified by the following $\F_{q^m}$-isomorphism:

\[
\begin{array}{lccl}
    \Psi \colon & \F_{q^m}^n & \longrightarrow & \F_{q^m}G \\
      & (u_1,\ldots,u_{n}) & \longmapsto & \sum_{i=1}^{n}u_ig_i,
\end{array}
\]
Via the isomorphism $\Psi$, we may endow the $\F_{q^m}$-vector space $\F_{q^m}^n$ with an algebra structure, by defining the product of two elements $u,v \in \F_{q^m}^n$ in the following way:
\[
uv:=\Psi^{-1}(\Psi(u) \Psi(v)).
\]

\begin{definition}
Let $a \in \F_{q^m}G$. The \textbf{left ideal matrix} generated by $a$ is the matrix
\[
\mathcal{LIM}(a):=\begin{pmatrix}
\Psi^{-1}(g_1a) \\
\Psi^{-1}(g_2a) \\
\vdots \\
\Psi^{-1}(g_na)
\end{pmatrix} \in \F_{q^m}^{n \times n}
\]
\end{definition}

\begin{remark}
The image via $\Psi$ of the vector space generated by the lines of $\mathcal{LIM}(a)$ is the left ideal generated by $a$ in $\F_{q^m}G$.
\end{remark}

Note that the product of two vectors of $\F_{q^m}$ can be expressed as the usual product vector-matrix: for $u=(u_1,\ldots,u_n), v=(v_1,\ldots,v_n) \in \F_{q^m}^n$,
\begin{align*}
    uv & = \Psi^{-1}(\Psi(u) \Psi(v))
=\Psi^{-1}\left(\displaystyle\sum_{i=1}^{n} u_ig_i\Psi(v)\right) 
 =\displaystyle\sum_{i=1}^{n} u_i \Psi^{-1} \left(g_i\Psi(v)\right) \\
&=(u_1,\ldots,u_n)\left(\begin{matrix}
\Psi^{-1}(g_1\Psi(v)) \\
\Psi^{-1}(g_2\Psi(v)) \\
\vdots \\
\Psi^{-1}(g_n\Psi(v))
\end{matrix}
\right)= u \mathcal{LIM}(\Psi(v)).
\end{align*}
Moreover, $\mathcal{LIM}(1)$ is clearly the identity.
We are now interested in invertible left ideal matrices. In this regard, since the group algebra is finite, if an element is left invertible, it is also right invertible and conversely. Moreover, the left inverse and the right inverses coincide.

\begin{proposition}[Proposition 9 of \cite{chimal2022efficient}]\label{prop:invert}
An element $a \in \F_{q^m}G$ is invertible if and only if the left ideal matrix $\mathcal{LIM}(a)$ is invertible.  
\end{proposition}

The investigations of units in finite group algebra is a subject of intense research (see for example 
\cite{sandling1984units,sahai2019unit,makhijani2014unit} and references therein).

\section{Low-rank parity-check codes from group algebras}

The LRPC codes have been introduced in \cite{gaborit2013lowrank}. This
family of codes can be seen as the equivalent of classical LDPC codes for the rank
metric. There is a natural analogy between low
density matrices and matrices with low rank.

\begin{definition}
Let $\lambda,K,N$ be positive integers with $0<K<N$. A \textbf{low-rank parity-check code} (LRPC) with parameter $(\lambda,K,N)$ is a code with a parity-check matrix $H$ such that $H=(h_{ij}) \in \F_{q^m}^{(N-K) \times N}$ is a full-rank matrix such that its coefficients generate an $\F_q$-subspace $\mathcal{F}$ of dimension $\lambda$, i.e. $\dim_{\F_q}(\langle h_{ij} \colon i\in \{1,\ldots,N-K\},j\in\{1,\ldots,N\} \rangle_{\F_q})=\lambda$.
\end{definition}

For ``small" values of $\lambda$, LRPC codes admit efficient probabilistic decoding algorithm (RSR) to recover the support of the error vector. More precisely, given as input an $\F_q$-basis of $\mathcal{F}$, $s=(s_1,\ldots,s_{\lambda})=eH^{\top} \in \F_{q^m}^N$ a syndrome of an error $e \in \F_{q^m}^N$ whose components belong to an $\F_q$-subspace $\mathcal{E}$ of $\F_{q^m}$, then the RSR algorithm outputs $\mathcal{E}$. For more details see \cite{gaborit2013lowrank}.

We may describe a low-rank parity-check code via its parity-check matrix. In order to reduce the size of a representation of this code, we can introduce left ideal low-rank parity-check codes via left ideal matrices.

\begin{definition} \label{def:idealparitycheck}
Let $\mathcal{F}$ be an $\F_q$-subspace of dimension $\lambda$ of $\F_{q^m}$ and let $h_1,h_2 \in \F_{q^m}G$. Suppose that $\langle h_{i,j} \colon 1 \leq i \leq n, 1 \leq j \leq 2 \rangle_{\F_q}=\mathcal{F}$ where $\Psi^{-1}(h_1)=(h_{1,1},\ldots,h_{n,1}),\Psi^{-1}(h_2)=(h_{1,2},\ldots,h_{n,2})$ and either $h_1$ or $h_2$ is an invertible element of $\F_{q^m}G$. A LRPC code is called \textbf{left ideal} if it has a parity-check matrix
\[
H_{h_1,h_2}:=\begin{pmatrix}
\mathcal{LIM}(h_1)^\top \mathcal{LIM}(h_2)^\top
\end{pmatrix}\in \F_{q^m}^{n \times 2n}.
\]
\end{definition}

Since either $h_1$ or $h_2$ is invertible, by Proposition \ref{prop:invert} either $\mathcal{LIM}(h_1)$ or $\mathcal{LIM}(h_2)$ is invertible, so that $H_{h_1,h_2}$ is a full-rank matrix. 

\begin{theorem}
If $\mathcal{C}$ is a left ideal LRPC code with parity-matrix $H_{h_1,h_2}$, with $h_1,h_2 \in \mathbb{F}_{q^m}G$ and either $h_1$ or $h_2$ is invertible, then the parameters of $\mathcal{C}$ are $(\lambda,n,2n)$, where 
\[ \lambda=\dim_{\fq}(\langle h_{i,j} \colon i \in\{1,\ldots,n\}, j\in\{1,2\} \rangle_{\fq}), \]
where $\Psi^{-1}(h_1)=(h_{1,1},\ldots,h_{n,1})$ and $\Psi^{-1}(h_2)=(h_{1,2},\ldots,h_{n,2})$.
If $h_1$ is invertible, then the systematic parity-check matrix of $\mathcal{C}$ is $H_{1,h_2h_1^{-1}}=\begin{pmatrix}
I_n \ \mathcal{LIM}(h_2{h_1}^{-1})^\top
\end{pmatrix}$.
\end{theorem}
\begin{proof}
Note that the entries of the vectors $\Psi^{-1}(g_ih_j)$ are the same of those of $h_j$, for any $i$ and $j$, since the map
$x \in G \mapsto xg_i \in G$ is a permutation of $G$. This implies that the $\F_q$-span of the entries of $h_1$ and $h_2$ corresponds to $\mathcal{F}$ (that is the $\fq$-span of the entries of $H$).
Straightforward computations show the second part, which we omit for brevity.
\end{proof}

A way to hide the structure of a left ideal LRPC code is to reveal only its systematic parity-check matrix. 
\section{ROLLO-I using group algebra left ideal LRPC codes}

A Key-Encapsulation scheme KEM = (KeyGen; Encap; Decap) is a triple of probabilistic
algorithms together with a key space $\mathcal{K}$. The key generation algorithm KeyGen generates a
pair of public and secret key $(pk; sk)$. The encapsulation algorithm Encap uses the public key $pk$ to produce an encapsulation $c$ and a key $K \in \mathcal{K}$. Finally Decap, using the secret key $sk$ and an encapsulation $c$, recovers the key $K$ or fails and return error.

We adapt the Key-encapsulation scheme ROLLO I by using this left ideal LRPC codes constructed using group algebras.

\begin{description}

	\item[KeyGen:] 
	\begin{itemize}
		\item Pick randomly an $\F_q$-subspace $\mathcal{F}$ of  $\F_{q^m}$ with dimension $\lambda$ and pick $(x,y) \in \mathcal{F}^n \times \mathcal{F}^n$ such that the $\F_q$-span of the components of $x$ and $y$ is $\mathcal{F}$ and $\Psi(x)$ is invertible (such a choice of $x$ depends on the group algebra, but it is important to remark that a ``general'' element in a group algebra is invertible and that it is easy to check if an element is invertible or not). 
		\item Compute $\Psi(h)=\Psi(y)\Psi(x)^{-1}$.
		\item Define $pk=(h,G)$ and $sk=(x,y)$.
	\end{itemize}
	\item[Encap(pk):] 
	\begin{itemize}
		\item Pick randomly an $\F_q$-subspace $\mathcal{E}$ of  $\F_{q^m}$ with dimension $r$ and pick $(e_1,e_2) \in \mathcal{E}^n \times \mathcal{E}^n$ such that the $\F_q$-span of the components of $e_i$ is $\mathcal{E}$, for $i \in \{1,2\}$
		\item Compute $\Psi(c)=\Psi(e_1)+\Psi(e_2)\Psi(h)$. 
		\item Compute $K=Hash(\mathcal{E})$ and give as output $c$.
	\end{itemize}
	\item[Dec(sk):]
	\begin{itemize}
		\item Compute $\Psi(c)\Psi(x)=\Psi(e_1)\Psi(x)+\Psi(e_2)\Psi(y) $, since $\Psi(h)=\Psi(y)\Psi(x)^{-1}$ and recover $\mathcal{E}$ with the algorithm RSR.
		\item $K=Hash(\mathcal{E})$.
	\end{itemize}
\end{description}

The above KEM scheme relies on the supposed hardness of following problem. 

\begin{problem}
Given an element $h \in \F_{q^m}G$, it is ``hard" to distinguish whether the code $\mathcal{C}$ with the parity-check matrix $H=\begin{pmatrix}
I_n \ \mathcal{LIM}(h)^\top
\end{pmatrix}$ is a random left ideal code or if it is an left ideal LRPC code with parameter $(\lambda,n,2n)$. \\
In other words, it is ``hard" to distinguish if $h$ was sampled uniformly at random or as $h_2{h_1}^{-1}$ where the $\F_q$-span of the components of the vectors $h_1$ and $h_2$ has dimension $\lambda$.  
\end{problem}

The large number of invertible elements in general group algebras and the fact that in general no canonical generators of ideals are known seems to suggest that the problem is likely to be hard.
Note that if $G$ is a cyclic group our construction of left ideal codes coincide with the the notion of quasi-cyclic code of index $2$. In this case, there is a structural attack against this construction (see \cite{hauteville2015new}). This implies that we have to suppose that $G$ is not cyclic and suggests also to avoid the abelian case. If $G$ is not cyclic, left ideal codes are similar to quasi-$G$ codes of index $2$ (see \cite{borello2019algebraic}), even if the transposition prevents them from being properly. It would be interesting to analyse if a similar attack exists also in this case. 

It is possible also to readapt the Cryptosystems ROLLO II and ROLLO III, using this approach.

\vspace{-3mm}

\section*{Acknowledgments}

The last two authors were supported by the project ``VALERE: VAnviteLli pEr la RicErca" of the University of Campania ``Luigi Vanvitelli'' and by the Italian National Group for Algebraic and Geometric Structures and their Applications (GNSAGA - INdAM).

\bibliographystyle{abbrv}
\bibliography{biblio}

\bigskip

\par\noindent Martino Borello\\
Universit\'e Paris 8, Laboratoire de G\'eom\'etrie, Analyse et Applications, LAGA,
Universit\'e Sorbonne Paris Nord, CNRS, UMR 7539, France\\
{{\em martino.borello@univ-paris8.fr}}

\bigskip
\par\noindent Paolo Santonastaso and Ferdinando Zullo\\
Universit\`a degli Studi della Campania ``Luigi Vanvitelli'', Caserta, Italy\\
{{\em \{paolo.santonastaso,ferdinando.zullo\}@unicampania.it}}

\end{document}